\documentclass[12pt]{article}
\usepackage{epsfig}
\usepackage{amsthm}
\usepackage{amsfonts}
\usepackage{color}
\usepackage{bm}
\usepackage{multirow}
\usepackage{pdfsync}
\usepackage{amssymb}
\usepackage{amsmath}
\usepackage{float}

\font\tenopen=msbm10
\font\sevenopen=msbm7
\font\fiveopen=msbm5
\newfam\openfam
\def\open{\fam\openfam\tenopen}
\textfont\openfam=\tenopen
\scriptfont\openfam=\sevenopen
\scriptscriptfont\openfam=\fiveopen
\font\title=cmbx10 scaled\magstep1

\def\R{{\open \mathbb{R}}}

\def\I{{\open \mathbb{I}}}

\theoremstyle{remark} \theoremstyle{lemma} \theoremstyle{definition}
\theoremstyle{corol} \theoremstyle{proposition}
\theoremstyle{condition} \theoremstyle{conjecture}
\newtheorem{theorem}{\bf{Theorem}}

\newtheorem{lemma}{\bf{Lemma}}

\newtheorem{example}{\bf{Example}}
\newtheorem{proposition}{\bf{Proposition}}

\makeatletter

\newcommand{\Rmnum}[1]{\expandafter\@slowromancap\romannumeral #1@}
\makeatother

\begin{document}
	
    \centerline{\bf{Universally Optimal Designs for the Two-dimensional Interference Model}}

\begin{center}
\begin{tabular}{cc}
\begin{tabular}{c}
    A. S. Hedayat\\
	Department of Math, Stat, and CS\\
	University of Illinois at Chicago \\
\end{tabular}	
	\bigskip
\begin{tabular}{c}	
	Heng Xu and Wei Zheng\\
	Department of Mathematical Science\\
	Indiana University-Purdue University Indianapolis
\end{tabular}
\end{tabular}
\end{center}

\bigskip

\abstract{There have been some major advances in the theory of optimal designs for interference models. However, the majority of them focus on one-dimensional layout of the block and the study for two-dimensional interference model is quite limited partly due to technical difficulties. This paper tries to fill this gap. Specifically, it systematically characterizes all possible universally optimal designs simultaneously. Computational issues are also addressed with theoretical backup.}
	
	\bigskip
	
	\centerline{\today}
	
	\bigskip
	
	\pagenumbering{arabic}
	
	\section{Introduction}

It is not uncommon in the application of block designs that a treatment assigned to a particular plot could have the so called neighbor or side effects on the neighboring plots. In avoiding systematic bias caused by these side effects, the interference model has gained its popularity in data analysis. Correspondingly, the optimal or efficient designs have been studied by Gill (1993), Druilhet (1999), Kunert and Martin (2000), Filipiak and Markiewicz (2003, 2005, 2007), Bailey and Druilhet (2004), Ai et al. (2007), Ai et al. (2009), Kunert and Mersmann (2011), Druilhet and Tinsson (2012) and Filipiak (2012), Li, Zheng and Ai (2015), Zheng (2015) and Zheng, Ai and Li (2017) among others. However, they all assumed the block to be in one-dimensional layout so that the side effects is only contributed by left and right neighbors. Not infrequently, many practical applications enforces the layout of blocks to be two dimensional so that the side effect applies to all four directions. See Langton (1990), Federer and Basford (1991), Morgan and Uddin (1991) and Williams, John and Whitaker (2006) for examples. This paper provides tools for characterizing optimal designs for a two-dimensional interference model. 

For the one-dimensional interference model, a design is essentially a collection of sequences of treatments. Similarly, a design for the two-dimensional interference model consists of many two-dimensional arrays. However, the change of dimension complicates the problem of finding optimal designs tremendously. As a result, the relevant study of optimal or efficient designs is quite limited. Langton (1990) proposed neighbour balanced Latin square without referring to a specific model. Federer and Basford (1991) constructed and compared three types of row-column designs with consideration of the side effects. Morgan and Uddin (1991) studied optimal designs at the presence of a particular spatial correlation structure without interference effects in the mean model. The latter work was followed by Uddin and Morgan (1997a, 1997b) and Morgan and Uddin (1999). We shall establish optimality conditions for the interference model for any spatial correlation structure. 

The paper is organized as follows. Section \ref{sec:notation} introduces the notations and formulates the problem. Section \ref{sec: sym} proposes a complete class and derives a necessary and sufficient condition for a design within it to be universally optimal. The condition leads to an explicit way of deriving the optimal or efficient designs. This section also provides some preliminary results useful for the proof of theorems in other sections. Section \ref{sec: linear} establishes a necessary and sufficient condition for an arbitrary design to be universally optimal. Section \ref{sec: seq} derives theoretical results regarding the supporting set of block arrays. This shrinks the pool of feasible designs and saves the computational cost tremendously. Section \ref{sec:example} provides some examples of optimal or efficient designs for various situations.

\section{Notations and formulations}\label{sec:notation}
	
    Consider a field experiment with $t$ treatments and $n$ blocks, and each block has $a$ rows and $b$ columns. Without loss of generality we assume $a \leq b$ since it doesn't change the mathematical form of the problem by switching the roles of row and column. The response at the $i$th row and $j$th column of block $k$ can be modeled as:
	\begin{eqnarray}
	y_{{ijk}}=\mu+\beta_k+\tau_{d(i,j,k)}+\gamma_{d(i-1,j,k)}+\gamma_{d(i+1,j,k)}+\gamma_{d(i,j-1,k)}+\gamma_{d(i,j+1,k)}+\varepsilon_{ijk},\label{eqn:800}
	\end{eqnarray}
	where the error term $\varepsilon_{ijk}$ has mean zero. The subscript $d(i,j,k)$ denotes the treatment assigned to position $(i,j)$ of block $k$ by the design $d:\{1,2,\cdots,a\}\times\{1,2,\cdots,b\}\times\{1,2,\cdots,n\}\to \{1,2,\cdots,t\}$. Here, $\mu$ is the average mean, $\beta_k$ is the block effect, $\tau_{d(i,j,k)}$ is the direct effect of treatment $d(i,j,k)$. Similarly, $\gamma_{d(i-1,j,k)},$ $\gamma_{d(i+1,j,k)},$ $\gamma_{d(i,j-1,k)},$ and $\gamma_{d(i,j+1,k)}$ are the side effects of treatments $d(i-1,j,k), d(i+1,j,k), d(i,j-1,k),$ and $d(i,j+1,k)$ from below, above, left and right plots, respectively. Here, we assume the side effect depends on the treatment only and does not depend on the direction. Suppose $Y_d$ is the vector of $y_{{ijk}}$ ordered by colexicographical order, then Model (\ref{eqn:800}) can be written in the matrix form of
	\begin{eqnarray}
	Y_d &=&1\mu+U\beta+T_d^{0}\tau+F_d\gamma+\varepsilon, \label{eqn:801}\\
	F_d &=& T_d^{1}+T_d^{2}+T_d^{3}+T_d^{4}\nonumber
	\end{eqnarray}
	where $\beta=(\beta_1,\cdots,\beta_n)^{'}$, $\tau=(\tau_1,\cdots,\tau_t)^{'}$, $\gamma=(\gamma_1,\cdots,\gamma_t)^{'}$ and $U=I_n\otimes 1_{p}$ with $p=a b$. Here $\otimes$ represents the Kronecker product, $1_{p}$ represents a vector of ones with length $p$, $I_n$ represents the identify matrix of size $n$ and $'$ means the transpose of a vector or a matrix. Also, $T_d^{0}$ and $T_d^{h}, 1\leq h\leq 4$, are the design matrices for the direct effect as well as the side effects from left, right, above, and below directions, respectively. We assume there is no guard plots or edge effects, i.e. $\gamma_{d(0,j,h)}=\gamma_{d(a+1,j,h)}=\gamma_{d(i,0,h)}=\gamma_{d(i,b+1,h)}=0$. Since the observations in $Y_d$ are organized by the colexicographical order, we have the decomposition $T_{d}^{i}=({T^{i}_1}',{T^{i}_2}',...,{T^{i}_n}')'$, $0\leq i\leq 4$, where $T^{i}_h$, $1\leq h\leq n$, is the incidence matrix of side effect treatment from each direction and block $h$. Further, we have $T_h^{1}=(I_b\otimes K_a)T_h^{0}$,
	$T_h^{2}=(I_b\otimes K_{a}^{'})T_h^{0}$,
	$T_h^{3}=(K_b\otimes I_a)T_h^{0}$,
	$T_h^{4}=(K_{b}^{'}\otimes I_a)T_h^{0}$, where $K_h=( \I_{[i-j=1]} )_{1\leq i, j \leq h}$ with $\I$ being the indicator function. That is, $F_d=MT_d^{0}$ with $M=I_n\otimes [(K_b\otimes I_a)+(K_{b}^{'}\otimes I_a)+(I_b\otimes K_a)+(I_b\otimes K_{a}^{'}))]$. 
	
	Regarding the dependence structure of the observations, we only adopt the very mild assumption $Var(\varepsilon)=I_n\otimes \Sigma$, where $\Sigma$ is a positive definite within-block covariance matrix. By Kunert (1984), the information matrix for $\tau$ is
	\begin{eqnarray}
	C_d&=&C_{d00}-C_{01}C_{d11}^{-}C_{d10},\label{eqn:802}
	\end{eqnarray}
	where $C_{d00}=T_{d}^{0'}(I_{n}\otimes \widetilde{B})T_d^{0}$, $C_{d01}=C_{d10}^{'}=T_{d}^{0'}(I_{n}\otimes \widetilde{B})F_d$, $C_{d11}=F_{d}^{'}(I_{n}\otimes \widetilde{B})F_d$ and 	$\widetilde{B} =\Sigma^{-1}-\Sigma^{-1}J_{p}\Sigma^{-1}(1_{p}^{'}\Sigma^{-1}1_{p})^{-1}$. The information matrix $C_d$ depends on the covariance matrix $\Sigma$ through the symmetric matrix $\tilde{B}$, whose row sum is zero. For the special case of $\Sigma=I_p$, we have the simplification of $\tilde{B}=B_p$, where $B_p:=I_p-p^{-1}J_p$. Kushner (1997) pointed out that when $\Sigma$ is of $type$-$H$, i.e. $\Sigma=xI_p+y1_p'+1_py'$ with $x\in \R$ and $y\in \R^k$, we have $\tilde{B}=B_p/x.$ Hence the choices of designs agree with that for $\Sigma=I_p$. This special case will be particularly dealt with in Section \ref{sec: seq}. We allow $\Sigma$ to be an arbitrary covariance matrix throughout the rest of the paper.

	To save the space, we represent a block array in the format $(t(\cdot,1);t(\cdot,2);...;t(\cdot,b))$, where $t(\cdot,j)=\{t(1,j),t(2,j),...,t(a,j)\}$ is the collection of treatments from the $j$th column of the block and $t(i,j)\in \{1,2,3,...,t\}$ is the treatment assigned to the $i$th row and $j$th column of the block. Hence, a design can be viewed as a result of selecting $n$ elements with replacement from ${\cal S}$, the set of all possible $t^{p}$ arrays. For an array $s\in\mathcal{S}$, let $n_s$ be the number of its replications in the design $d$ and $p_s=n_s/n$ be the proportion of it. When $n$ is fixed, a design is determined by the $measure$ $\xi=\{p_s,s\in{\cal S}\}\in {\cal P}$, where ${\cal P}=\{\xi|\sum_{s\in{\cal S}}p_s=1$, $p_s\geq 0\}$. Implicitely, we have dropped the requirement that $np_s$ has to be an integer for all $s\in {\cal S}$. This relaxation allows us to solve the optimization problem through calculus tools. Essentially, for any design, its associated measure shall be in the space of ${\cal P}$. The derived solutions not only provides a benchmark for measuring the efficiency of any exact design, but also guides us to derive optimal or efficient designs.
	
	Now we shall demonstrate that the search for optimal design can be approached by searching for optimal measure. For $0\leq i, j \leq1$, let $C_{sij}$ be the degenerated matrix of $C_{dij}$ when design $d$ consists of a single array $s$. Note that matrices $C_{dij}$, $0\leq i,j\leq 1$ are additive in the blocks, namely $C_{dij}=\sum^n_{h=1}C_{hij}$ with $C_{hij}={T_h^{(i)}}'\tilde{B}T_h^{(j)}$. Suppose $\xi=\{p_s, s\in {\cal S}\}$ is the measure associated with the design $d$, then we have $C_{dij}=nC_{\xi ij}$, where $C_{\xi ij}=\sum_{s\in{\cal S}}p_sC_{sij}$. As a result, we have
	\begin{eqnarray}
	C_d&=&nC_{\xi}\label{eqn:512}\\
	C_{\xi}&=&C_{\xi 00}-C_{\xi 01}C_{\xi 11}^{-}C_{\xi 10}, \label{eqn:measure}
	\end{eqnarray}
	
	Equation (\ref{eqn:512}) shows that the maximization of $C_d$ can be achieved by maximizing $C_{\xi}$. Follwing Kiefer (1975), a measure $\xi$ is said to be $universally$ $optimal$ if it maximizes $\Phi(C_{\xi})$ for any $\Phi$ satisfying  the following three conditions. \\
	$(C.1)$ $\Phi$ is concave.\\
	$(C.2)$ $\Phi$ is nondecreasing.\\
	$(C.3)$ $\Phi(S^{'}CS)=\Phi(C)$ for any permutation matrix $S$.\\

	\section{The complete class}\label{sec: sym}
	In approximate design theory, one powerful tool is the complete class introduced in the seminal book by Karlin and Studden (1966), on Chebyshev systems. It tries to identify a subset of simple structured designs which at the same time contains the optimal design. As a result, we can easily find the optimal design within this complete class. Some general related theories have been developed in a series of papers by Yang and Stufken (2009), Yang (2010), Dette and Melas (2011), Yang and Stufken (2012) and Dette and Schorning (2013). Unfortunately, the methodologies based on Chebyshev system does not apply here since the design point is multi-dimensional and constrained within a discrete domain. However, the symmetrization idea adopted by Kushner (1997) in his study of optimal crossover design applies to our problem.
	
	Let ${\cal G}$ be the set of all $t!$ permutations on symbols $\{1,2,...,t\}$. For permutation $\sigma\in {\cal G}$ and array $s$, we define $\sigma s$ to be the array derived by applying the permutation $\sigma$ to each element of $s$, that is, the $(i,j)$th element of $\sigma s$ is $\sigma [t(i,j)]$. We call a measure to be {\it symmetric} if $p_s=p_{\sigma s}$ for all $s\in {\cal S}$ and $\sigma\in {\cal G}$. For array $s$, denote by $\langle s\rangle=\{\sigma s:\sigma\in {\cal G}\}$ the {\it symmetric block set (SBS)} generated by $s$. As ${\cal G}$ assembles a group in abstract algebra, we have the partition ${\cal S}=\cup^m_{i=1}\langle s_i\rangle$, where $m$ is the number of distinct SBS's which partition ${\cal S}$. Let $p_{\langle s_i\rangle}=\sum_{s\in \langle s_i\rangle}p_s$ be the {\it SBS proportion} and $|\langle s_i\rangle|$ be the cardinality of $\langle s_i\rangle$. Then, for a symmetric measure, we shall have
	\begin{eqnarray}\label{eqn:0325}
	p_s=p_{\langle s_i\rangle}/|\langle s_i\rangle|&for& s\in \langle s_i\rangle, 1\leq i\leq m.
	\end{eqnarray}
	That is, the SBS proportion $p_{\langle s_i\rangle}$ is evenly allocated to each array in the corresponding SBS. 
	
	
	\begin{lemma}\label{lemma:0223}
	There exists a symmetric measure which is universally optimal among ${\cal P}$. 
	\end{lemma}
	
	\begin{proof}
	For any measure $\xi=\{p_s,s\in{\cal S}\}\in {\cal P}$ and permutation $\sigma\in {\cal G}$, let $\xi_{\sigma}=\{p_{\sigma^{-1}s},s\in{\cal S}\}$
    and $\xi^{*}=(\sum_{\sigma\in {\cal G}}\xi_{\sigma})/t!$. $\xi^{*}$ satisfies (\ref{eqn:0325}). 
    
    Since $C_{\xi^{*}ij}=\sum_{\sigma\in {\cal G}}(C_{\xi_{\sigma}ij})/t!$, $0 \leq i, j \leq 1$, by the concaveness of Schur's complement, we have $C_{\xi^{*}}\geq \sum_{\sigma\in {\cal G}}C_{\xi_{\sigma}}/t!$, which together with conditions (C.1)-(C.3) yield $\Phi(C_{\xi^{*}})\geq \Phi(C_{\xi})$.
	\end{proof}

	Lemma \ref{lemma:0223} has identified the collection of all symmetric measures to be a complete class. Next, we will show the information matrix of a symmetric measure is of a very simple format so that the maximization of it becomes tractable. Let $c_{\xi ij}=tr(B_tC_{\xi ij})$, $0 \leq i,j \leq 1$, so that we have $c_{\xi ij}=\sum_{s\in {\cal S}} p_sc_{sij}$, where $c_{sij}=tr(B_tC_{sij})$. For a symmetric measure $\xi$, one can verify that $C_{\xi ij}$, $0\leq i, j\leq1$ is completely symmetric and $C_{\xi00}$, $C_{\xi01}$ have zero column sums. Hence we have
	\begin{eqnarray}
	C_{\xi ij}=c_{\xi ij}B_t/(t-1)+\mathbb{I}_{[i=j=1]}1_{t}^{'}C_{\xi ij}1_{t}J_t/t^2,\label{eqn:803}
	\end{eqnarray}
	Regarding this representation, we have
	\begin{lemma}\label{positive}
	For any array $s\in {\cal S}$, we have $c_{s11} > 0$ and $1_{t}^{'}C_{s11}1_{t}>0$.
	\end{lemma}
	\begin{proof}
	From its definition, $C_{s11}$ is non-negative definite. If  $c_{s11}=0$, we will have $C_{s11}=F_{s}^{'}\widetilde{B}F_s=0$, which implies $\widetilde{B}F_s=0$ and thus $F_s1_{t}=1_{p}v$ where $v$ is a scaler.  Similarly, if $1_{t}^{'}C_{s11}1_{t}=0$, we would also have $F_s1_{t}=1_{p}v$. However, $F_s1_{t}=1_{p}v$ is not possible by the structure of $F_s$. 
	\end{proof}

	By Lemma \ref{positive}, we have $c_{\xi11} > 0$ and $1_{t}^{'}C_{\xi11}1_{t}>0$ for any measure due to the linearity relationship, and hence $C_{\xi11}$ is positive definite for any symmetric measure.
	Denote $q_{\xi}^{*}=(c_{\xi00}-c_{\xi01}^{2}/c_{\xi11})$. By direct calculations, we have
	\begin{eqnarray}
	C_{\xi}(\tau)&=&q_{\xi}^{*}B_t/(t-1),\label{eqn:804}
	\end{eqnarray}
	for any symmetric measure $\xi$ in view of (\ref{eqn:measure}). Recall the partition ${\cal S}=\cup^m_{k=1}\langle s_k\rangle$ and note that $c_{sij}$ is the same for arrays from the same SBS. We shall have the representation: $c_{\xi ij}=\sum^m_{k=1} p_{\langle s_k\rangle}c_{s_kij}$, $0 \leq i,j \leq 1$. All these together with (\ref{eqn:804}) leads to a convenient way of constructing a universally optimal design: Find the proper SBS proportion $p_{\langle s_k\rangle}, 1\leq k\leq m$, so as to maximize $q_{\xi}^{*}$ and then allocate the SBS proportion uniformly to each individual sequence within the SBS. 
	
	In fact, we shall be able to enlarge the complete class to all measures for which the matrix $C_{\xi ij}$, $0 \leq i,j \leq 1$, is completely symmetric. We call such measure as {\it pseudo symmetric} in order to distinguish it from the already defined notion of symmetric measures. In fact, one can easily verify that a symmetric measure is always pseudo symmetric and also (\ref{eqn:804}) holds for all pseudo symmetric measures. Given the optimal SBS proportions, its associated pseudo symmetric measure should also be universally optimal design among all measures in ${\cal P}$. The following proposition provides more details of what we have concluded so far.
	\begin{proposition}\label{prop:optimal}
		Let $y^{*}=max_{\xi \in {\cal P}}q_{\xi}^{*}$. $(i)$ A pseudo symmetric measure is universally optimal if and only if $q_{\xi}^{*}=y^{*}$. $(ii)$ Recall $B_t=I_t-J_t/t$. A measure is universally optimal if and only if $C_{\xi}=y^{*}B_t/(t-1)$.
	\end{proposition}
	
	Part $(ii)$ of Proposition \ref{prop:optimal} is due to the concavity argument given by Kiefer (1975). This is the corner stone for deriving the optimality condition for asymmetric measures as in Section \ref{sec: linear}. On the other hand, part $(i)$ indicates that it suffices to maximize $q_{\xi}^{*}$ if the consideration is confined to psuedo symemtric measures. Note that the computational complexity for maximizing $q_{\xi}^{*}$ is generally $O(m^3)$, where $m$ is the number of distinct SBS's and could grow very fast as the size of design increases. Now we introduce two different results, each leading to significant save of computational time.

	\begin{theorem}\label{thm:513}
	Let $q_s(x)=c_{s00}+2c_{s01}x+c_{s11}x^2$ for $x\in \R$. A psuedo symemtric measure $\xi$ is universally optimal under Model (\ref{eqn:800}) if and only if
		\begin{eqnarray}\label{eqn:513}
		\min_{s\in {\cal S}}q_s\left(\frac{c_{\xi 01}}{c_{\xi 11}}\right)&=&q^*_{\xi}.
		\end{eqnarray}
	If $\xi$ is not universally optimal, we have $\min_{s\in {\cal S}}q_s\left(\frac{c_{\xi 01}}{c_{\xi 11}}\right)<q^*_{\xi}$.
	\end{theorem}
	
	Theorem \ref{thm:513} is of the Kiefer's type equivalence theorem, and can be easily derived by the traditional method of using Fr\'echet derivative in view of the fact that $q_{\xi}^{*}$ is a concave functional of the measure $\xi$. Condition (\ref{eqn:513}) not only helps check the optimality of a measure but also provides the guideline of improving on a non-optimal measure. The well known Federov's exchange algorithm can be easily adopted here to achived the maximum of $q^*_{\xi}$. The computational complexity of maximizing $q_{\xi}^{*}$ by using Theorem \ref{thm:513} is only $O(m)$.
	
	Alternatively, Kushner (1997) has derived another type of optimality condition through the quadratic function $q_s(x)$ as defined in Theorem \ref{thm:513}. By examining the arguments therein, it can be veried that we can have a similar result. To save the space, we shall only provide the results without proof. Let $r(x)=\max_{\xi}q_\xi(x)$, then by Lemma \ref{positive}, $r(x)$ is a strictly convex function with a unique minimizer which is denoted by $x^*$ here. Further, we have $y^*=r(x^*)$. Let ${\cal Q}=\{s\in {\cal S}|q_s(x^*)=y^*\}$ be the collection of arrays pathing through $(x^*,y^*)$. Then we have
	\begin{theorem}\label{corol}
		A pseudo symmetric design is universally optimal under Model (\ref{eqn:800}) if and only if
		\begin{eqnarray}
		\sum_{s\in{\cal Q}}p_s(c_{s01}+x^{*}c_{s11})&=&0,\label{eqn:12222}\\
		p_s&=&0,\mbox{ if }s\notin {\cal Q}.\label{eqn:12224}
		\end{eqnarray}
	\end{theorem}
	Condition (\ref{eqn:12224}) shows that ${\cal Q}$ contains all supporting arrays for any universally optimal psuedo symmetric measure. In Section \ref{sec: linear} we shall show that this is true for any measure. Condition (\ref{eqn:12222}) means that we only need to solve a simple linear equation to derive the optimal proportion. Kushner (1997) suggested finding $x^*$ through pairwise comparison among all SBS pairs and hence the computational complexity is $O(m^2)$ accordingly. As shown in the proof of Lemma \ref{sym:point}, $x^*$ could be derived once a universally optimal measure is derived. By relying on Theorem \ref{thm:513}, we can reduce the complexity of deriving $x^*$ and hence ${\cal Q}$ back to $O(m)$. One advantage of Theorem \ref{corol} is that it helps derive all possible universally optimal measures simutanously.

	Even though the optimality conditions given in Theorems \ref{thm:513} and \ref{corol} appears so different, they indeed cover the same set of designs since both of them are the necessary and sufficient conditions for a pseudo symmetric measure to be universally optimal. Hedayat and Zheng (2017) discussed the construction of psuedo symmetric measures in the study of crossover designs, and the adoption of orthogonal array of type I ($OA_I$) therein still applies here. We shall illustrate the idea through examples in section \ref{sec:example}.
	
	\section{The general optimality condition}\label{sec: linear}
	Section \ref{sec: sym} established the optimality condition for measures in the complete class of psuedo symmetric measures. This section shall characterize universally optimal measures in the whole class ${\cal P}$. Let $\mathcal{V}_{\xi}=\{s: p_s>0, s\in \mathcal{S}\}$ be the support of $\xi$. Lemma \ref{sym:point} shows that the set of arrays ${\cal Q}$ defined earlier contains the support of any universally optimal measure. Theorem \ref{optimal:eqn} shows that one can characterize all the universally optimal measures by a system of linear equations regarding the array sequences $p_s$, $s\in {\cal Q}$.
	\begin{lemma}\label{sym:point}
	$(i)$ If ${\xi}$ is universally optimal, we have $q_{\xi}^{*}=y^*$, which further indicates $x^{*}=-c_{\xi01}/c_{\xi11}$ and $\mathcal{V}_{\xi}\subset {\cal Q}$.
	\end{lemma}

	\begin{proof}
	Let $q_{\xi}(x)=\sum_{s\in {\cal S}}p_sq_s(x)$, we would have $q_{\xi}(x)=c_{\xi00}+2c_{\xi01}x+c_{\xi11}x^2$. We can verify that $q_{\xi}^{*}=min_{x\in \mathbb{R}}q_{\xi}(x)$ and the minimum is achieved if and only if $x=-c_{\xi01}/c_{\xi11}$. By (5.3) in Kushner (1997) we have $tr(C_{\xi}(\tau))\leq tr(C_{\xi00}) +2tr(C_{\xi01})x+ tr(C_{\xi11}B_t)x^2$ for all $x\in R$. Now set $x=-c_{\xi 01}/c_{\xi 11}$, we have $tr(C_{\xi}(\tau))\leq q_{\xi}^{*}\leq y^*$. As a result we have $q_{\xi}^{*}=y^*$ in view of Proposition \ref{prop:optimal}. Note that the unique minimizer of $q_{\xi}(x)$ is $\tilde{x}=-c_{\xi01}/c_{\xi11}$. If $x^*\neq \tilde{x}$ then $y^*=r(x^*)\geq q_{\xi}(x^*)> q_{\xi}(\tilde{x})=q_{\xi}^{*}$, contradicted. If there is an array, say $s$, with $s \in \xi$ and $s \notin {\cal Q}$, we have $y^*>q_{\xi}(x^*)\geq q_{\xi}^{*}$ and hence the contradiction is reached. 
	\end{proof}

	\begin{theorem}\label{optimal:eqn}
		A measure $\xi$ is universally optimal under Model (\ref{eqn:800}) if and only if
		\begin{eqnarray}
		\sum_{s\in{\cal Q}}p_s(C_{s00}+x^{*}C_{s01})&=& y^{*}B_t/(t-1),\label{optimal1} \\
		\sum_{s\in{\cal Q}}p_s(C_{s10}+x^{*}C_{s11})&=& 0,\label{optimal2} \\
		p_s &=& 0,\mbox{ if }s\notin {\cal Q}.\label{optimal3}
		\end{eqnarray}
	\end{theorem}
	
	\begin{proof}
	First (\ref{optimal3}) is a direct result of Lemma \ref{sym:point}. Let $\xi^{'}$ be a symmetric optimal measure and $\xi^{*}=\xi/2+\xi^{'}/2$, $\xi^{*}$ is also universally optimal since $\Phi$ and the Schur complement are both concave. By the same argument of Theorem 5.3 in Kushner (1997)  we have
	\begin{eqnarray}
	0 &=& C_{\xi^{'}11}(-C_{\xi^{*}11}^{+}C_{\xi^{*}10}+C_{\xi^{'}11}^{+}C_{\xi^{'}10})\label{merge1}, \\
	0 &=& C_{\xi11}(-C_{\xi^{*}11}^{+}C_{\xi^{*}10}+C_{\xi11}^{+}C_{\xi10})\label{merge2},
	\end{eqnarray}
	where $^{+}$ means the Moore-Penrose generalized inverse. By (\ref{eqn:803}) and Lemma \ref{positive}, $C_{\xi^{'}11}$ is non-singular, which together with (\ref{merge1}) and Lemma \ref{sym:point} implies $C_{\xi^{*}11}^{+}C_{\xi^{*}10}=-x^{*}B_t$. Then we get (\ref{optimal2}) by (\ref{optimal3}) and (\ref{merge2}).\\
	By (\ref{optimal2}), (\ref{optimal3}) and Proposition \ref{prop:optimal}, we have
	\begin{eqnarray}
	y^{*}B_t/(t-1) &=& C_{\xi00}-C_{\xi01}C_{\xi11}^{+}C_{\xi10},\nonumber\\
	&=& C_{\xi00}+x^{*}C_{\xi01}\label{result2} .
	\end{eqnarray}
	which together with (\ref{optimal3}) implies (\ref{optimal1}).\\
	The sufficiency of (\ref{optimal1}) - (\ref{optimal3}) is straightforward in view of (\ref{result2}).
	
	\end{proof}

	\section{Theoretical form of ${\cal Q}$ when $\Sigma$ is of type-H}\label{sec: seq}
	The major challenge with the two dimensional interference model is that the number of block arrays could be very large. For a small design with $t=3$ treatments and the block size of $a\times b=3\times 3$ , there are $3^9=19,683$ block arrays and $m=3,281$ distinct SBS. It is crucial to have knowledge of the supporting set of arrays,  ${\cal Q}$, before resorting to computer. In Section \ref{sec: sym}, we have argued that ${\cal Q}$ can be derived within the time complexity of $O(m)$. However, $m$ could also be large as the design size continues to grow. Continuing with the previous example, just by increasing the value of $b$ from $3$ to $4$, the value of $m$ is increased from $3,281$ to $88,574$. In this section, we will give the theoretical value and form of $x^*$ and ${\cal Q}$ for all feasible combinations of $a$, $b$ and $t$ when $\Sigma$ is of type-H. 
	
	A treatment is said to be {\it significant} in a block array if it appears twice in adjacent plots and also one of its replications is on a corner of the block. If both plots assigned to the treatment is on corners of the array, it's said to be $strictly$ $significant$. Of course, this is only possible when $a=2$. Recall that we assume $a \leq b$ throughout the paper without loss of generality. Let ${\cal Q}_i$, $0\leq i\leq 4$, be the collection of block arrays, where there are $i$ significant treatments and $ab-2i$ treatments replicated exactly once. For $1\leq j\leq2$, let ${\cal Q}_j^*$ be a subset of ${\cal Q}_j$ such that all significant treatments are strictly significant. Particularly, ${\cal Q}_0$ represents the collection of all binary block arrays, for which no treatment is replicated for more than once. At last, let ${\cal Q}^*=\{s: |f_{s,i}-f_{s,j}|\leq 1, 1\leq i,j\leq t\}$. Theorems \ref{t=p-2}--\ref{t=p} provide the theoretical form of the supporting set  ${\cal Q}$ for cases of $t\leq p-2$,  $t=p-1$ and $t\geq p$, respectively. The proofs of them are tedious and hence deferred to the appendix.
	
	\begin{theorem}\label{t=p-2}
		Under Model (\ref{eqn:800}) with $\Sigma=I_{ab}$ and $t\leq p-2$, we have:
		\begin{eqnarray}
			x^*&=&0,\\
			y^*&=&p-\frac{p^2+r(t-r)}{pt},\\
			{\cal Q}&=&{\cal Q}^*,\label{Q: p-2}
		\end{eqnarray}
		where $r$ is the remainder obtained by dividing $p$ by $t$, $f_{s,i}$ is the number of replications of treatment $i$ in block array $s$.
	\end{theorem}

	\begin{theorem}\label{t=p-1}
		Under Model (\ref{eqn:800}) with $\Sigma=I_{ab}$ and $t=p-1$. \\
		$(I)$If $a \geq 3$, we have
		\begin{eqnarray}
		x^*&=&\frac{p-(a+b-5/2)}{\eta p-(16p-14a-14b+20)}\label{eqn:x-star},\\
		y^*&=&p-\frac{p+2}{p}+ 2(\frac{2a+2b-5}{p}-2)x^*+(\eta-\frac{16p-14a-14b+20}{p})x^{*2}\label{eqn:y-star},\\
		{\cal Q} &=& {\cal Q}_1,
		\end{eqnarray}
		where $\eta=4p-2a-2b-2(8ab-7a-7b+4)/t+4(2p-a-b)^{2}/(pt)$.\\	
	    $(II)$ If $a=2$ and $b\geq 3$, we have $x^*=(\eta+6/b-9)^{-1}$, $y^*=2b-(b+1)/b-2x^*+(\eta+6/b-9)x^{*2}$, ${\cal Q} = {\cal Q}_1^*$. \\
	    $(III)$ If $a=b=2$, we have  $x^*=1/2$, $y^*=2$ and ${\cal Q} = \cup_{j=1}^{2}{\cal Q}_j^*$. 
	\end{theorem}

	\begin{theorem}\label{t=p}
		Under Model (\ref{eqn:800}) with $\Sigma=I_{ab}$ and $t\geq p$.\\
		$(I)$ If $a \geq 3$, we have
		\begin{eqnarray}
		x^*&=&\frac{(2p-5)-\sqrt{(2p-5)^2-24}}{12},\label{ir x}\\
		y^*&=&p-\frac{p+2}{p}+ 2(\frac{2a+2b-5}{p}-2)x^*+(\eta-\frac{16p-14a-14b+20}{p})x^{*2}\label{y-star:2},\\
		{\cal Q}&=&\mathcal{M}:=\bigcup_{i=0}^{4}{\cal Q}_i.\label{Q: p}
		\end{eqnarray}
		$(II)$ If $a= 2$ and $b\geq 3$, we have $x^*=(b-1-\sqrt{(b-1)^2-1})/2$, $y^*=2b-1+ (4/b-6)x^*+(\eta-9+10/b)x^{*2}$ and ${\cal Q}={\cal Q}_0\bigcup {\cal Q}_1^*\bigcup {\cal Q}_2^*$.\\
		$(III)$ If $a=b=2$, we have  $x^*=1/2$, $y^*=2$ and ${\cal Q} = {\cal Q}_0\bigcup {\cal Q}_1^*\bigcup {\cal Q}_2^*$.
	\end{theorem}

\section{Examples}\label{sec:example}
   
This section illuminates the theorems of this paper through some concrete examples. We shall mainly focus on the case when $\Sigma$ is of type-H since we have theoretical form of ${\cal Q}$ and $x^*$ given in Section \ref{sec: seq}. But in general, it is matter of quick computational search based on results from Sections \ref{sec: sym} and \ref{sec: linear}. Specifically, based on Theorem \ref{thm:513}, we can build a Federov's type of exchagne algorithm to derive a measure which maximizes $q_{\xi}^*$. With this measure, we can have $x^*=\arg\min_{x\in \R} q_{\xi}(x)$ and $y^*=q_{\xi}(x^*)$. The set ${\cal Q}$ can hence be obtained by its definition. To this point, there are two ways of deriving optimal or efficient designs. One is to find a proper value of the SBS proportions based on Theorem \ref{corol} and then construct a symmetric (by full permutation) or a pseudo symmetric design (by using $OA_I$). This method needs $n$ to be a multiple of a certain number, see Hedayat and Zheng (2017) for further details. The other is to translate the linear equations in Theorem \ref{optimal:eqn} into an integer quadratic programming problem and try to give an optimal or efficient design for an arbitrary value of $n$, see Zheng (2013) as an example in finding optimal or efficient crossover designs. Also, we shall mention that both methods apply to all combinations of $a,b,t$ and $\Sigma$.

To evaluate the performance of a design, we need to define its statistical efficiency. Let $0\leq \lambda_1 \leq \lambda_2 \leq... \leq \lambda_{t-1}$ be the $t$ eigenvalues of $C_d$ for a design $d$, then we define A-, D- and E- and T-efficiencies of $d$ as follows. 
\begin{eqnarray*}
\varepsilon_{A}(d)&=&\frac{(t-1)^2}{ny^{*}\left(\sum_{i=1}^{t-1}\lambda_i^{-1}\right)},\\
\varepsilon_{D}(d)&=&\frac{t-1}{ny^{*}}  \left(\Pi_{i=1}^{t-1}\lambda_i\right)^{1/(t-1)},\\
\varepsilon_{E}(d)&=&\frac{(t-1)\lambda_1}{ny^{*}},\\
\varepsilon_{T}(d)&=&\frac{\sum_{i=1}^{t-1}\lambda_i}{ny^{*}}.
\end{eqnarray*}
We can see that a (pseudo) symmetric design should have the identical value of efficiency under different criteria. Also, a universally optimal design can be verified to have unity efficiency under those four criteria. Following the structure of Section \ref{sec: seq}, we shall present the examples based on three different caes, namely $t\leq p-2$, $t= p-1$ and $t\geq p$. Some designs in literature are also included in the comparison.

\subsection{The case of $t\leq p-2$}
\begin{example}\label{ex1}
Suppose $(a,b,t)=(2,3,2)$. Based on Theorem \ref{corol}, a symmetric exact design with  $p_{\langle(1,\; 2;\; 2,\; 1;\; 1,\; 2)\rangle}=\frac{1}{8}$ and $p_{\langle(1,\; 1;\; 2,\; 1;\; 2,\; 2)\rangle}=\frac{7}{8}$ will be universally optimal and the minimum value of $n$ should be $16$ to have such design. Meanwhile, by Theorem \ref{optimal:eqn}, we are able to construct universally optimal design with only $n=4$ block arrays as follows.
\[
\begin{bmatrix}
   1&1&2\\
   1&2&2
\end{bmatrix}
\begin{bmatrix}
   1&1&2\\
   1&2&2
\end{bmatrix}
\begin{bmatrix}
   1&1&2\\
   2&1&2
\end{bmatrix}
\begin{bmatrix}
   1&2&1\\
   2&2&1
\end{bmatrix}
\]
\end{example}

\begin{example} Suppose $(a,b,t)=(5,5,5)$. We can find the following design in Langton (1990), who has proposed neighbour balanced Latin square. 
\[s_1=
\begin{bmatrix}
   1&2&3&4&5\\
   4&5&1&2&3\\
   2&3&4&5&1\\
   5&1&2&3&4\\
   3&4&5&1&2   
\end{bmatrix}
\]
A (pseudo) symmetric design based on $s_1$ yields the efficiency of $0.5151$. It sounds unfair to include $s_1$ in the comparison since Langton (1990) actually did not target on any particular model or parameter. Here, we only try to use $s_1$ as the starting point to construct an effcient design. Note that $s_1$ is contained in ${\cal Q}$ in view of (\ref{Q: p-2}). Now define $s_2$, 
\[s_2=
\begin{bmatrix}
   1&2&3&4&5\\
   1&2&3&4&5\\
   1&2&3&4&5\\
   1&2&3&5&4\\
   1&2&3&4&5  
\end{bmatrix}
\]
then a (pseudo) symmetric design with $p_{\langle s_1\rangle}=p_{\langle s_2\rangle}=1/2$ is actually universally optimal.
\end{example}

\begin{example}\label{ex4}
Suppose $(a,b,t)=(6,8,4)$. Similar to Langton (1990), Chan and Eccleston(1998) proposed the following array without referring to any specific model.  
\[s_3=
\begin{bmatrix}
   1&2&4&3&4&1&3&2\\
   1&2&1&4&2&3&1&4\\
   4&1&3&2&1&2&4&3\\
   3&4&2&1&2&3&1&4\\
   4&1&3&2&4&1&3&2\\
   2&3&1&4&3&4&2&1
\end{bmatrix}
\]
The efficiency of a (pseudo) symmetric design based on $s_1$ is 0.6821. Again, $s_3$ is contained in ${\cal Q}$ in view of (\ref{Q: p-2}). Let 
\[s_4=
\begin{bmatrix}
   1&1&2&2&3&3&4&4\\
   1&1&2&2&3&3&4&4\\
   1&1&2&2&3&3&4&4\\
   1&1&2&2&3&3&4&4\\
   1&1&2&2&3&4&4&4\\
   1&1&2&2&3&3&3&4
\end{bmatrix}
\]
Then a (pseudo) symmetric design with $2p_{\langle s_3\rangle}=p_{\langle s_2\rangle}=2/3$ yields the efficiency of $0.9999$.
\end{example}

\subsection{The case of $t= p-1$}
\begin{example}\label{ex2}
Suppose $(a,b,t)=(2,3,5)$. A (pseudo) symmetric design by using only one SBS, $\langle(1\; 1;\; 2\; 3;\; 4\; 5)\rangle$, is universally optimal, and the minimum value of $n$ for such designs is $20$.
\end{example}

\begin{example}
Suppose $(a,b,t)=(3,3,8)$. A (pseudo) symmetric design by using only one SBS, $\langle(1\; 1\; 2;\; 3\; 4\; 5;\; 6\; 7\; 8\;)\rangle$, is universally optimal, and the minimum value of $n$ for such designs is $56$.
\end{example}

\begin{example}
Suppse $(a,b,t)=(3,4,11)$. A (pseudo) symmetric design by using only one SBS, $\langle(1\; 1\; 2;\; 3\; 4\; 5;\; 6\; 7\; 8;\; 9\; 10\; 11)\rangle$, is universally optimal, and the minimum value of $n$ for such designs is $110$.
\end{example}

\subsection{The case of $t\geq p$}
When $t \geq p$, $x^{*}$ is typically irrational for most combinations of $a$ and $b$ according to (\ref{ir x}). Consequence, universally optimal exact designs rarely exist. However, we are able to construct highly efficient designs for any combination of $a$ and $b$. For example, when $(a,b,t)=(2,3,6)$, a (pseudo) symmetric design by using only one SBS, $\langle(1\; 1;\; 2\; 3;\; 4\; 5)\rangle$, yields efficiency of 0.9997. When $(a,b,t)=(3,4,12)$, a (pseudo) symmetric design by using only one SBS, $\langle(1,\; 1,\; 2;\; 3,\; 4,\; 5 ;$ $6,\; 7,\; 8;\; 9,\; 10,\;11)\rangle$, yields efficiency of 0.9999. In fact, highly efficient (pseudo) symmetric designs can always be constructed based on block arrays in ${\cal Q}_1$ or ${\cal Q}_{1}^{*}$. Figure 1 shows the high efficiencies of such designs under different combinations of $a$,$b$,and $t$. In Example 8, we focus on constructing an efficient asymmetric design for an arbitrary number of $n$.
\begin{figure}[H]
\centering
\includegraphics[scale=0.555]{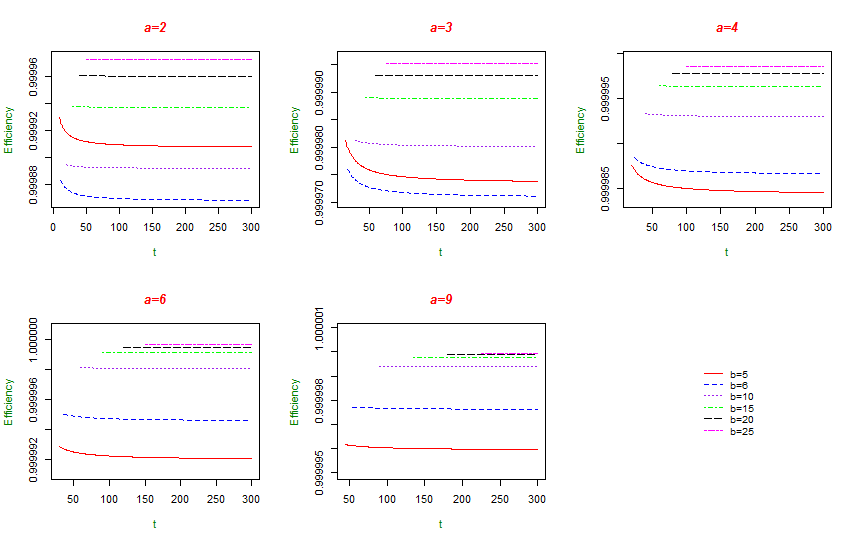}
\caption{Efficiency of pseudo symmetric designs under
various combinations of $a$, $b$, $t$.}
\end{figure}

\begin{example}\label{ex5}
Suppose $(a,b,t)=(4,2,8)$. Uddin and Morgan(1997b) gave the following design. 
\[
\begin{bmatrix}
   8&1\\
   5&7\\
   3&4\\
   6&2
\end{bmatrix}
\begin{bmatrix}
   8&2\\
   6&1\\
   4&5\\
   7&3
\end{bmatrix}
\begin{bmatrix}
   8&3\\
   7&2\\
   5&6\\
   1&4
\end{bmatrix}
\begin{bmatrix}
   8&4\\
   1&3\\
   6&7\\
   2&5
\end{bmatrix}
\begin{bmatrix}
   8&5\\
   2&4\\
   7&1\\
   3&6
\end{bmatrix}
\begin{bmatrix}
   8&6\\
   3&5\\
   1&2\\
   4&7
\end{bmatrix}
\begin{bmatrix}
   8&7\\
   4&6\\
   2&3\\
   5&1
\end{bmatrix}
\]
\[
\begin{bmatrix}
   2&3\\
   7&8\\
   1&5\\
   4&6
\end{bmatrix}
\begin{bmatrix}
   3&4\\
   1&8\\
   2&6\\
   5&7
\end{bmatrix}
\begin{bmatrix}
   4&5\\
   2&8\\
   3&7\\
   6&1
\end{bmatrix}
\begin{bmatrix}
   5&6\\
   3&7\\
   4&1\\
   8&2
\end{bmatrix}
\begin{bmatrix}
   6&7\\
   4&8\\
   5&2\\
   1&3
\end{bmatrix}
\begin{bmatrix}
   7&1\\
   5&8\\
   6&3\\
   2&4
\end{bmatrix}
\begin{bmatrix}
   1&2\\
   6&8\\
   7&4\\
   3&5
\end{bmatrix}
\]
Their model is slightly different in the sense that the side effects are not included in the mean part but the two dimensional layout is accounted by a particular within-block covariance. Note that each block array of this design is contained in ${\cal Q}$ in view of (\ref{Q: p}). The efficiency of it is $\varepsilon_{A}(d)=0.9750$, $\varepsilon_{D}(d)=0.9754$, $\varepsilon_{E}(d)=0.9134$ and $\varepsilon_{T}(d)=0.9759$. By Theorems \ref{optimal:eqn} and \ref{t=p}, we can also construct an alternative design. 
\[
\begin{bmatrix}
   1&1\\
   2&8\\
   3&7\\
   6&4
\end{bmatrix}
\begin{bmatrix}
   6&6\\
   8&1\\
   3&7\\
   5&4
\end{bmatrix}
\begin{bmatrix}
   2&2\\
   5&7\\
   3&1\\
   8&4
\end{bmatrix}
\begin{bmatrix}
   7&7\\
   2&3\\
   1&4\\
   6&8
\end{bmatrix}
\begin{bmatrix}
   5&5\\
   2&8\\
   3&7\\
   6&1
\end{bmatrix}
\begin{bmatrix}
   4&4\\
   2&1\\
   5&3\\
   6&8
\end{bmatrix}
\begin{bmatrix}
   8&7\\
   4&6\\
   2&3\\
   5&1
\end{bmatrix}
\]
\[
\begin{bmatrix}
   8&8\\
   5&1\\
   7&4\\
   2&6
\end{bmatrix}
\begin{bmatrix}
   3&3\\
   1&4\\
   5&6\\
   7&2
\end{bmatrix}
\begin{bmatrix}
   4&5\\
   2&8\\
   3&7\\
   6&1
\end{bmatrix}
\begin{bmatrix}
   5&6\\
   3&8\\
   4&1\\
   7&2
\end{bmatrix}
\begin{bmatrix}
   6&7\\
   4&8\\
   5&2\\
   1&3
\end{bmatrix}
\begin{bmatrix}
   7&1\\
   5&8\\
   6&3\\
   2&4
\end{bmatrix}
\begin{bmatrix}
   1&2\\
   6&8\\
   7&4\\
   3&5
\end{bmatrix}
\]
The efficiency of it is $\varepsilon_{A}(d)=0.9792$, $\varepsilon_{D}(d)=0.9806$, $\varepsilon_{E}(d)=0.9002$ and $\varepsilon_{T}(d)=0.9820$. As mentioned earlier, $x^{*}$ is irrational due to (\ref{ir x}), hence a universally optimal exact design does not exist anyway. That indicates both Uddin and Morgan(1997b)'s and our design performs reasonably well here.
\end{example}

\section{Appendix}

This section proves the results in Section \ref{sec: seq}. We would like to briefly explain the structure of this section.  Lemma \ref{core parameters} calculates the coefficients of $q_s(x)$ as defined in Theorem \ref{thm:513}, which is repeatedly needed in the rest of this section. Lemmas \ref{lemma 11}--\ref{lemma 33} are technical results for proving Theorem \ref{t=p-1} and Lemmas \ref{lemma 111}--\ref{lemma 333} are technical results for proving Theorem \ref{t=p}. 

To proceed, we shall define some technical notations. Given an array $s$, recall that $t(i,j)$ is the treatment at the $(i,j)$th location. For treatment $m$, define $f_{s,m}^{1}=\sum_{i=1}^{a}\sum_{j=1}^{b-1}\mathbb{I}_{[t(i,j)=m]}$,
 $f_{s,m}^{2}=\sum_{i=1}^{a}\sum_{j=2}^{b}\mathbb{I}_{[t(i,j)=m]}$,
 $f_{s,m}^{3}=\sum_{i=1}^{a-1}\sum_{j=1}^{b}\mathbb{I}_{[t(i,j)=m]}$,
 $f_{s,m}^{4}=\sum_{i=2}^{a}\sum_{j=1}^{b}\mathbb{I}_{[t(i,j)=m]}$. They are numbers of replications of treatment $m$ in various areas of the array. Here, there is an abuse of the notation $m$. In the previous sections, it represents the total number of SBS and here it represents a treatment index. With $h_{s}^{i,j} = \sum_{m=1}^{t}f_{s,m}^{i}f_{s,m}^{j},$ $0\leq i, j\leq 4$, we further define
 $h_{s}^{1} = \sum_{j=1}^{4}h_{s}^{0,j}$,
 $h_{s}^{2} = \sum_{i=1}^{4}h_{s}^{i,i}$,
 $h_{s}^{3} = \sum_{1\leq i< j \leq 4}h_{s}^{i,j}$.
 Also, let $\rho_s=\sum_{m=1}^{t}\mathbb{I}_{[f_{s,m}^{0}>0]}$ be the number of different treatments in array $s$, and 
 $\mathcal{N}_i=\{s\in \mathcal{S}: \rho_s=i\}$, $1\leq i \leq p$, be the collection of arrays in which there are $i$ different treatments. The following notations are merely technical without obvious interpretations.

\begin{center}
\begin{tabular}{ c c }
 $z_{s,r,m}^{1} =\sum_{i=1}^{a}\sum_{j=1}^{b-1}\mathbb{I}_{[t(i,j)=t(i,j+1)=m]}$ & $z_{s,r,m}^{2} =\sum_{i=1}^{a}\sum_{j=1}^{b-2}\mathbb{I}_{[t(i,j)=t(i,j+2)=m]}$\\ 
 $z_{s,c,m}^{1} =\sum_{i=1}^{a-1}\sum_{j=1}^{b}\mathbb{I}_{[t(i,j)=t(i+1,j)=m]}$ & $z_{s,c,m}^{2} =\sum_{i=1}^{a-2}\sum_{j=1}^{b}\mathbb{I}_{[t(i,j)=t(i+2,j)=m]}$\\  
 $z_{s,r}^{1} = \sum_{m=1}^{t}z_{s,r,m}^{1}$ & $z_{s,r}^{2} =\sum_{m=1}^{t}z_{s,r,m}^{2}$\\ 
  $z_{s,c}^{1} = \sum_{m=1}^{t}z_{s,c,m}^{1}$ & $z_{s,c}^{2}= \sum_{m=1}^{t}z_{s,c,m}^{2}$\\
 $z_{s,d,m}^{1}=\sum_{i=2}^{a}\sum_{j=1}^{b-1}\mathbb{I}_{[t(i,j)=t(i-1,j+1)=m]}$ & $z_{s,d,m}^{2}=\sum_{i=1}^{a-1}\sum_{j=1}^{b-1}\mathbb{I}_{[t(i,j)=t(i+1,j+1)=m]}$\\
 $z_{s,d}^{1}=\sum_{m=1}^{t}z_{s,d,m}^{1}$ & $z_{s,d}^{2}=\sum_{m=1}^{t}z_{s,d,m}^{2}$\\
 $z_{s}^{1} = 2z_{s,r}^{1}+2z_{s,c}^{1}$ & $z_{s}^{2} = 2z_{s,r}^{2}+2z_{s,c}^{2}+4z_{s,d}^{1}+4z_{s,d}^{2}$
\end{tabular}
\end{center}

\begin{lemma}\label{core parameters}
Given an array $s$, we have $c_{s00} = p-h_{s}^{0,0}/p$, $c_{s01} = z_{s}^{1}-h_{s}^{1}/p$ and $c_{s11} = \eta+z_{s}^{2}-h_{s}^{2}/p-2h_{s}^{3}/p$. Recall, $\eta=4p-2a-2b-2(8ab-7a-7b+4)/t+4(2p-a-b)^{2}/(pt)$
\end{lemma}
\begin{proof}
By direct calculations, we have
\begin{eqnarray*}
c_{s00}&=&tr(T_{s}^{'}B_pT_{s})=tr(T_{s}^{'}T_{s})-tr(T_{s}^{'}1_p1_{p}^{'}T_s)/p=p-h_{s}^{0,0}/p,\\
c_{s01}&=&tr(T_{s}^{'}B_pF_{s})=\sum_{1\leq i\leq 4}tr(T_{s}^{0'}B_pT_{s}^i)=\sum_{1\leq i\leq 4}\left[tr(T_{s}^{0'}T_{s}^i)-tr(T_{s}^{0'}1_p1_{p}^{'}T_{s}^i)/p\right]\\
&=&2z_{s,r}^{1}+2z_{s,c}^{1}-(\sum_{1\leq i\leq 4}h_{s}^{0,i})/p=z_{s}^{1}-h_{s}^{1}/p,\\
c_{s11} &=& tr(F_{s}^{'}B_pF_{s}B_t)=\sum_{1\leq i\leq 4}tr(T_{s}^{i'}B_pT_{s}^{i}B_t)+2\sum_{1\leq i< j \leq 4}tr(T_{s}^{i'}B_pT_{s}^{j}B_t).
\end{eqnarray*}
To complete the calculation for $c_{s11}$, we will examine one term in the above expression.
\begin{eqnarray*}
&&tr(T_{s}^{1'}B_pT_{s}^{3}B_t)\\
&=&tr(T_{s}^{1'}T_{s}^{3})-tr(T_{s}^{1'}1_p1_{p}^{'}T_{s}^3)/p-tr(T_{s}^{1'}T_{s}^{3}1_t1_{t}^{'})/t+tr(T_{s}^{1'}1_p1_{p}^{'}T_{s}^31_t1_{t}^{'})/(pt)\\
&=&z_{s,d}^{1}-h_{s}^{1,3}/p-(a-1)(b-1)/t+(p-a)(p-b)/(pt).
\end{eqnarray*}
\end{proof}

\begin{proof}[\textbf{Proof of Theorem \ref{t=p-2}}]
It is sufficient to verify that $(i)$  $\max_{s\in {\cal S}}q_s(0)=p-[p^2+r(t-r)]/(pt)$ and $\mathcal{Q}^{*}$ is the set of arrays at which the maximum is achieved. $(ii)$ $\partial q_{s}(x)/\partial x|_{x=0}=c_{s01}\geq 0$ for some $s \in \mathcal{Q}^{*}$. $(iii)$ $\partial q_{s}(x)/\partial x|_{x=0}=c_{s01}<0$ for some other $s \in \mathcal{Q}^{*}$.

For the special case of $a=b=2$, it is straightforward to verify $(i)$--$(iii)$. Particularly, the sequences satisfying the conditions $(ii)$ and $(iii)$ are $(1,2;2,1)$ and $(1,1;2,2)$ respectively. In the sequel, we consider $b\geq 3$. Part $(i)$ is straightforward in view of the facts  $\mathcal{Q}^{*}=\{s: |f_{s,i}^0-f_{s,j}^0|\leq 1, 1\leq i, j \leq t\}$ and $q_s(0) = p-h_{s}^{0,0}/p$.

Part $(ii)$. Given an array $s$ and treatment $m$, define $\Lambda_{s,m}=\{(i,j): t(i,j)=m\}$. We say treatment $m$ is $connected$ in $s$ if for any $(i,j), (i',j')\in \Lambda_{s,m}$, there exists an array of  positions $(i_1,j_1),(i_2,j_2),\cdots,(i_l,j_l)\in \Lambda_{s,m}$ such that $(i,j)=(i_1,j_1)$, $(i',j')=(i_l,j_l)$ and $|i_k-i_{k-1}|+|j_k-j_{k-1}|\leq1$ for $2 \leq k\leq l$. By convention, we also call treatment $m$ to be $connected$ if $|\Lambda_{s,m}|\leq 1$. Let $\Lambda^{*}=\{(1,1),(1,b),(a,1),(a,b)\}$, we show that part $(ii)$ is satisfied by any array $s$ $\in$ $\mathcal{Q}^{*}$, such that all treatments in $s$ are connected, $f_{s,m}^{0}=q+1$ for $1\leq m\leq r$ and particularly $\Lambda_{s,m}\cap\Lambda^*\neq\phi$ for m=1,2. Recall that $c_{s01} = z_{s}^{1}-h_{s}^{1}/p$. By induction, one can show that $z_{s,c,m}^{1}+z_{s,r,m}^{1} \geq f_{s,m}^{0}-1$.
Hence $z_{s}^{1}=2\sum_{m=1}^{t}(z_{s,c,m}^{1}+z_{s,r,m}^{1})\geq 2(p-t)$, together with the fact that $p=\sum_{m=1}^{t}f_{s,m}^{0}$, we have $c_{s01}\geq\sum_{m=1}^{t}f_{s,m}^{0}(2-\sum_{j=1}^{4}f_{s,m}^{j}/p)-2t:= f$. Next, we will show $f\geq0$ in three separate cases. Case $(a)$, $r\geq 2$. We have $f\geq \sum_{m=1}^{t}q(2-\sum_{j=1}^{4}f_{s,m}^{j}/p)+\sum_{m=1}^{2}(2-\sum_{j=1}^{4}f_{s,m}^{j}/p)-2t=2(q-1)(t-2)+q(2a+2b)/p-2(0+1+2+2)/p\geq0$. Case $(b)$, $r=1$. Since $p\geq 6$ and $p-t\geq 2$, we have $q\geq 2$ and $t+q\geq 5$, thus $f= \sum_{m=1}^{t}q(2-\sum_{j=1}^{4}f_{s,m}^{j}/p)+\sum_{m=1}^{1}(2-\sum_{j=1}^{4}f_{s,m}^{j}/p)-2t=2(q-1)(t-2)+q(2a+2b)/p-5/p-2>0$. Case $(c)$, $r=0$. Since $p\geq 6$ and $p-t\geq 2$, we have $q\geq 2$ and $t+q\geq 5$, thus $f=2(q-1)(t-2)-4+(2a+2b)/t>0$.

Part $(iii)$. Recall $c_{s01} = z_{s}^{1}-h_{s}^{1}/p$ by Lemma \ref{core parameters} and notice the fact that $h_s^1>0$ for any array, part $(iii)$ will be verified if we can find an array $s$ such that $z_{s}^{1}=0$. Recall $z_{s}^{1}=2\sum_{m=1}^{t}\sum_{i=1}^{a}\sum_{j=1}^{b-1}\mathbb{I}_{[t(i,j)=t(i,j+1)=m]}+2\sum_{m=1}^{t}\sum_{i=1}^{a-1}\sum_{j=1}^{b}\mathbb{I}_{[t(i,j)=t(i+1,j)=m]}$. Hence $z_{s}^{1}=0$ could be achieved by any array in which no treatment is assigned to any neighboring plots of the block array.




\end{proof}

\begin{proof}[\textbf{Proof of Theorem \ref{t=p-1}}]
We shall only prove the theorem for the case of $a\geq 3$, since the case of $a=2$ follows from a similar but a lot simpler argument. Let $x^{*}=(p-a-b+5/2)/(\eta p-16p+14a+14b-20)$, it is sufficient to show that $(i)$ Given an array $s\in {\cal Q}_1$, $x=x^{*}$ is the minimizer of $q_s(x)$. $(ii)$ $\max_{s\in {\cal}}q_s(x^*)$ equal to the right hand side of (\ref{eqn:y-star}) and the maximum is achieved by an array if and only if it belongs to $\mathcal{Q}_{1}$.

By Lemma \ref{core parameters}, we have $c_{s00}= p-(p+2)/p$, $c_{s01}  = (2a+2b-5)/p-2$ and $c_{s11} =\eta-(16p-14a-14b+20)/p$ for $s\in {\cal Q}_1$. Then the verification of $(i)$ is left to some simple algebra. Part $(ii)$ is a direct result of Lemmas \ref{lemma 11}, \ref{lemma 22} and \ref{lemma 33}.
\end{proof}

\begin{proof}[\textbf{Proof of Theorem \ref{t=p}}]
Similar to the Proof of Theorem \ref{t=p-1}, here we shall only give the proof for the case of $a\geq 3$. Let $x^*=(2p-5-\sqrt{(2p-5)^2-24})/12$, it is sufficient to show that $(i)$ When $x=x^{*}$, there exist $s_1, s_2 \in {\cal M}$, such that $\partial q_{s_1}(x)/\partial x|_{x=x^{*}} > 0 > \partial q_{s_2}(x)/\partial x|_{x=x^{*}}$. $(ii)$ When $x=x^{*}$, the value of $y^*$ in (\ref{y-star:2}) is the maximum value of $q_s(x)$ and ${\cal M}$ is the set of arrays at which the maximum is achieved.  

Part $(ii)$ is a direct result of Lemmas \ref{lemma 111}, \ref{lemma 222} and \ref{lemma 333}.  In the sequel, we shall focus on part $(i)$. Let $s_1 \in {\cal Q}_0$ and $s_2 \in {\cal Q}_2$, by Lemma \ref{core parameters}, we have
$c_{s_100}=p-1$, $c_{s_101}=-(4p-2a-2b)/p$, $c_{s_111}=\eta-(16p-14a-14b+8)/p$, $c_{s_200}=p-(p+2)/p$, $c_{s_201}=(2a+2b-10)/p$ and $c_{s_211}=\eta-(16p-14a-14b+32)/p$. Let $x_0=-c_{s_101}/c_{s_111}$ and $x_2=-c_{s_201}/c_{s_211}$, by Lemma \ref{positive}, it is sufficient to show that $x_0<x^{*}<x_2$ or equivalently:
\begin{eqnarray}
\frac{(-2a-2b+10)/p}{\eta-(16p-14a-14b+32)/p}&<&\frac{(2p-5)-\sqrt{(2p-5)^2-24}}{12}\label{Q0},\\
\frac{(2p-a-b)/p}{\eta-(16p-14a-14b+8)/p}&>&\frac{(2p-5)-\sqrt{(2p-5)^2-24}}{12}\label{Q1}.
\end{eqnarray}
(\ref{Q0}) follows by the fact that the left part is negative while it is positive on the right.
For (\ref{Q1}), let $\eta=h(t)$, then we have $\partial h(t)/\partial t=2t^{-2}[a^2(b-2)+b^2(a-2)](ab)^{-1}>0$, hence
\begin{eqnarray*}
\frac{(2p-a-b)/p}{\eta-(16p-14a-14b+8)/p}&\geq& \frac{(2p-a-b)/p}{h(\infty)-(16p-14a-14b+8)/p}\\
&>&\frac{(2p-5)-\sqrt{(2p-5)^2-24}}{12}
\end{eqnarray*}

\end{proof}

\begin{lemma}\label{lemma 11}
When $a\geq 3$ and $t=p-1$. Let $\ddot{\mathcal{Q}_1}=\{s\in \mathcal{S}, \rho_s=p-1\}$.
Then for any $s\in \mathcal{S}\setminus \ddot{\mathcal{Q}_1}$, there exists $s^*\in \ddot{\mathcal{Q}_1}$ such that $q_{s^*}(x^*)>q_{s}(x^*)$, where $x^*$ is given by (\ref{eqn:x-star}).
\end{lemma}
\begin{proof}
It is sufficient to show that for any array $s_1$ with $\rho_{s_1}<p-1$, there exists an array $s_2$ such that $\rho_{s_2}=\rho_{s_1}+1$ and $q_{s_2}(x^*)>q_{s_1}(x^*)$. 

For an array $s_1$ with $\rho_{s_1}<p-1$, we can always find a treatment $m_1$ and another treatment $m_2$,  such that $f_{s_1,m_1}^{0}\geq 2$ and $f_{s_1,m_2}^{0}=0$. Let $(i^{'},j^{'})\in \Lambda_{s_1,m_1}$, such that for any $(i,j)\in \Lambda_{s_1,m_1}$ we have  $i^{'}+j^{'}> i+j$, or $i^{'}+j^{'}=i+j$ with $i^{'}>i$. Let $s_2$ be the new array obtained from $s_1$ by setting $t(i^{'},j^{'})=m_2$ and others remain unchanged. By the definition of $\rho_s$, we have $\rho_{s_2}=\rho_{s_1}+1$. In the rest of the proof we show $q_{s_2}(x^*)>q_{s_1}(x^*)$ in separate cases.

Case $(i)$, $p\geq 24$. Recall $x^*=[1-1/a-1/b+5/(2p)]/(\eta -16+14/a+14/b-20/p)$. For the numerator of $x^*$, we have $0<1-1/a-1/b+5/(2p)<1$. For the denominator of $x^*$, we have $\eta-(16-14/a-14/b+20/p)
>4p-2a-2b-16/a-16/b+8/(ab)-(16-14/a-14/b+20/p)
>(4p-12/p)-(2a+2/a)-(2b+2/b)-16
>(a-2)(b-2)+3ab-23
>3p-22$.
Hence, $0<x^*<(3p-22)^{-1}$. Next, We shall compare the coefficients of $q_{s_2}(x^*)$ and $q_{s_1}(x^*)$.
By Lemma \ref{core parameters}, $c_{s00} = p-h_{s}^{0,0}/p$, $c_{s01} = z_{s}^{1}-h_{s}^{1}/p$ and $c_{s11} = \eta+z_{s}^{2}-h_{s}^{2}/p-2h_{s}^{3}/p$. One can verify that $h_{s_1}^{i}-h_{s_2}^{i}\geq 0$ for $1\leq i \leq 3$. For example, $h_{s_1}^{1}-h_{s_2}^{1}=\sum_{j=0}^{4}[(f_{s_1,m_1}^{j}-f_{s_2,m_2}^{j})+f_{s_2,m_1}^{0}(f_{s_1,m_1}^{j}-f_{s_2,m_1}^{j})] \geq 0$. Hence
\begin{eqnarray*}
c_{s_200}-c_{s_100}&=&2(f_{s_1,m_1}^{0}-1)/p\\
c_{s_201}-c_{s_101}&\geq& 2(z_{s_2,c,m_1}^{1}-z_{s_1,c,m_1}^{1})+ 2(z_{s_2,r,m_1}^{1}-z_{s_1,r,m_1}^{1})\\
c_{s_211}-c_{s_111}&\geq &2(z_{s_2,c,m_1}^{2}-z_{s_1,c,m_1}^{2})+2(z_{s_2,r,m_1}^{2}-z_{s_1,r,m_1}^{2})\\
&&+ 4(z_{s_2,d,m_1}^{1}-z_{s_1,d,m_1}^{1})+4(z_{s_2,d,m_1}^{2}-z_{s_1,d,m_1}^{2}).
\end{eqnarray*}
By the relationship between $s_1$ and $s_2$, we have $z_{s_2,o,m_1}^{i}-z_{s_1,o,m_1}^{i}\geq -1$, $i=1,2$, $o=c,r,d$. Now we are ready to show $q_{s_2}(x^*) -q_{s_1}(x^*)>0$.
Case $(a)$, $f_{s_1,m_1}^{0}= 2$. We have $c_{s_200}-c_{s_100}=2/p$,
$c_{s_201}-c_{s_101}\geq -2$ and 
$c_{s_211}-c_{s_111}\geq -4$, which together with $0<x^*<(3p-22)^{-1}$ yield $q_{s_2}(x^*) -q_{s_1}(x^*)>0$. Case $(b)$,  $f_{s_1,m_1}^{0}\geq3$. We have $c_{s_200}-c_{s_100}\geq 4/p$, $c_{s_201}-c_{s_101}\geq -4$ and
$c_{s_211}-c_{s_111}\geq -12$, which also gives $q_{s_2}(x^*) -q_{s_1}(x^*)>0$.

Case $(ii)$, $p<24$. There are only finite many combinations of $(a,b,t)$ and the values of $x^*$ as well as $q_{s_1}(x^*)$ and $q_{s_2}(x^*)$ can all be explicitly evaluated. We have verified the statement in this theorem for all these specific combinations.


\end{proof}

\begin{lemma}\label{lemma 22}
When $a\geq 3$ and $t=p-1$. Let $\dot{\mathcal{Q}_1}=\{s: s\in \ddot{\mathcal{Q}_1}$, all treatments are connected in s\}. For any $s\in \ddot{\mathcal{Q}_1}\setminus \dot{\mathcal{Q}_1}$ and $s^*\in \dot{\mathcal{Q}_1}$, we have $q_{s^*}(x^*)>q_{s}(x^*)$, where $x^*$ is given by (\ref{eqn:x-star}).
\end{lemma}
\begin{proof}
By a similar argument as in the proof of Lemma \ref{lemma 11}, we have $0<x^*<(3p-22)^{-1}$,
$c_{s^*00}-c_{s00}=0$, 
$c_{s^*01}-c_{s01}=z_{s^*}^{1}-z_{s}^{1}-(h_{s^*}^{1}-h_{s}^{1})/p\geq 2-4/p$ and 
$c_{s^*11}-c_{s11}=z_{s^*}^{2}-z_{s}^{2}-(h_{s^*}^{2}-h_{s}^{2})/p-2(h_{s_2}^{3}-h_{s_1}^{3})/p\geq -4-30/p$. Hence, $q_{s^*}(x^*)-q_{s}(x^*)\geq (-4-30/p)x^{*2}+2(2-4/p)x^*=x^*[4p-8-(4p+30)x^*]/p>0$
\end{proof}

\begin{lemma} \label{lemma 33}
When $a\geq 3$ and $t=p-1$.  For any $s\in \dot{\mathcal{Q}_1}\setminus \mathcal{Q}_1$ and  $s^*\in \mathcal{Q}_1$, we have $q_{s^*}(x^*)>q_{s}(x^*)$, where $x^*$ is given by (\ref{eqn:x-star}).
\end{lemma}

\begin{proof}
By a similar argument as in the proof of Lemma \ref{lemma 11}, we have $0<x^*<(3p-22)^{-1}$, $c_{s^*00}-c_{s00}=0$, $c_{s^*01}-c_{s01}=z_{s^*}^{1}-z_{s}^{1}-(h_{s^*}^{1}-h_{s}^{1})/p=-(h_{s^*}^{1}-h_{s}^{1})/p> 0$, and $c_{s^*11}-c_{s11}=z_{s^*}^{2}-z_{s}^{2}-(h_{s^*}^{2}-h_{s}^{2})/p-2(h_{s^*}^{3}-h_{s}^{3})/p=-(h_{s^*}^{2}-h_{s}^{2})/p-2(h_{s^*}^{3}-h_{s}^{3})/p> 0$. Hence, $q_{s^*}(x^*)-q_{s}(x^*)>0$.
\end{proof}

\begin{lemma} \label{lemma 111}
When $a\geq 3$ and$t\geq p$. Let $\mathcal{M}_2=\{s\in\mathcal{Q}: f_{s,m}^{0}\leq 2, 1\leq m\leq t\}$, then for any $s\in \mathcal{S}\setminus \mathcal{M}_2$, there exists $s^*\in \mathcal{M}_2$ such that $q_{s^*}(x^*)>q_{s}(x^*)$, where $x^*$ is given by (\ref{ir x}).
\end{lemma}
\begin{proof}

 For an array $s_1\in \mathcal{S}\setminus \mathcal{M}_2$, by the fact that $t\geq p$ we can always find a treatment $m_1$ and another treatment $m_2$,  such that $f_{s_1,m_1}^{0}\geq 3$ and $f_{s_1,m_2}^{0}=0$. Let $(i^{'},j^{'})\in \Lambda_{s_1,m_1}$, such that for any $(i,j)\in \Lambda_{s_1,m_1}$ we have  $i^{'}+j^{'} > i+j$, or $i^{'}+j^{'}=i+j$ while $i^{'}>i$. Let $s_2$ be the new array obtained from $s_1$ by letting $t(i^{'},j^{'})=m_2$ and others remain unchanged. Here, it is sufficient to show $q_{s_2}(x^*)>q_{s_1}(x^*)$. Case $(a)$, $f_{s_1,m_1}^{0}=3$. We have $c_{s_200}-c_{s_100}= 4/p$, $c_{s_201}-c_{s_101}\geq-4$ and $c_{s_211}-c_{s_111}\geq-8$. Case $(b)$, $f_{s_1,m_1}^{0}\geq4$. We have $c_{s_200}-c_{s_100}\geq 6/p$, $c_{s_201}-c_{s_101}\geq-4$ and $c_{s_211}-c_{s_111}\geq-12$. As a result, we have $q_{s_2}(x^*)>q_{s_1}(x^*)$.
\end{proof}

\begin{lemma}\label{lemma 222}
When $a\geq 3$ and $t\geq p$. Let $\mathcal{M}_1$=$\{s:s\in\mathcal{M}_2$, all treatments are connected in $s$\}.
For any $s\in \mathcal{M}_2\setminus \mathcal{M}_1$, there exists $s^*\in \mathcal{M}_1$ such that $q_{s^*}(x^*)>q_{s}(x^*)$, where $x^*$ is given by (\ref{ir x}).
\end{lemma}
\begin{proof}
Let $s_1\in \mathcal{M}_1\setminus \mathcal{M}_2$, since $t\geq p$, we can always find a treatment $m_1$ and another treatment $m_2$,  such that $m_1$ not connected, $f_{s_1,m_1}^{0}=2$ and $f_{s_1,m_2}^{0}=0$. Let $s_2$ be the new array obtained from $s_1$ by replacing one replication of $m_1$ to $m_2$ and others remain unchanged. Here, it is sufficient to show that  and $q_{s_2}(x^*)>q_{s_1}(x^*)$. We have $c_{s_200}-c_{s_100}=2/p$, $c_{s_201}-c_{s_101}=(h_{s_1}^{1}-h_{s_2}^{1})/p\geq 4/p$ and $c_{s_211}-c_{s_111}\geq z_{s_2}^{2}-z_{s_1}^{2}\geq-4$. Hence $q_{s_2}(x^*) -q_{s_1}(x^*)>-4x^{*2}+8x^*/p+2/p>0$.
\end{proof}

\begin{lemma}\label{lemma 333}
When $a\geq 3$ and $t\geq p$. For any $s\in \mathcal{M}_1\setminus \mathcal{M}$, there exists $s^*\in \mathcal{M}$ such that $q_{s^*}(x^*)>q_{s}(x^*)$, where $x^*$ is given by (\ref{ir x}).
\end{lemma}
\begin{proof}
Let $s_1\in \mathcal{M}_1\setminus \mathcal{M}$, then  by definition we can always find a treatment $m_1$ and another treatment $m_2$,  such that $f_{s_1,m_1}^{0}=2$, $\Lambda_{s,m_1}\cap\Lambda^*=\phi$ and $f_{s_1,m_2}^{0}=0$. Let $s_2$ be the new array obtained from $s_1$ by replacing one replication of $m_1$ to $m_2$ and others remain unchanged. Here, it is sufficient to show that $q_{s_2}(x^*)>q_{s_1}(x^*)$. We have $c_{s_200}-c_{s_100}=2/p$, 
$c_{s_201}-c_{s_101}\geq -2+6/p$ and 
$c_{s_211}-c_{s_111}=(h_{s_1}^{2}-h_{s_2}^{2})/p+2(h_{s_1}^{3}-h_{s_2}^{3})/p\geq 0$. Thus $q_{s_2}(x^*) -q_{s_1}(x^*)>2(6/p-2)x^*+2/p>0$.
\end{proof}

\bigskip
\bigskip

\bigskip

\noindent{\bf References}

\bigskip

\par\noindent\hangindent2.3em\hangafter 1
Ai, M., Ge, G. and Chan, L. (2007). Circular neighbor-balanced designs universally optimal for total effects. {\it Sci. China Ser.} A {\bf 50} 821--828. 

\par\noindent\hangindent2.3em\hangafter 1
Ai, M., Yu, Y. and He, S. (2009). Optimality of circular neighbor-balanced designs for total effects with autoregressive correlated observations. {\it J. Statist. Plann. Inference} {\bf 139} 2293--2304.

\par\noindent\hangindent2.3em\hangafter 1
Bailey, R. A. and Druilhet, P. (2004). Optimality of neighbor-balanced designs for total effects. {\it Ann. Statist.} {\bf 32} 1650--1661.

\par\noindent\hangindent2.3em\hangafter 1
Dette, H. and Schorning, K. (2013). Complete classes of designs for nonlinear regression models and principal representations of moment spaces. {\it Ann. Statist.}
{\bf 41} 1260--1267.

\par\noindent\hangindent2.3em\hangafter 1
Druilhet, P. (1999). Optimality of neighbour balanced designs. {\it J. Statist. Plann. Inference}  {\bf 81} 141--152.

\par\noindent\hangindent2.3em\hangafter 1
Druilhet, P. and Tinsson, W. (2012). Efficient circular neighbour designs for spatial interference model. {\it J. Statist. Plann. Inference} {\bf 142} 1161--1169.

\par\noindent\hangindent2.3em\hangafter 1
Gill, P. S. (1993). Design and analysis of field experiments incorporating local and remote effects of treatments. {\it Biom. J.} {\bf 35} 343--354.

\par\noindent\hangindent2.3em\hangafter 1
Hedayat, A. S., and Zheng, W. (2017). The story of symmetry in constructing crossover designs. {\it Manusript}.

\par\noindent\hangindent2.3em\hangafter 1
Federer, W. T. and Basford K. E. (1991). Competing effects designs and models for two-dimensional field arrangements. {\it International Biometric Society} {\bf 47} 1461--1472.

\par\noindent\hangindent2.3em\hangafter 1
Filipiak, K. (2012). Universally optimal designs under an interference model with equal left- and right-neighbor effects. {\it Statist. Probab. Lett.} {\bf 82} 592--598.

\par\noindent\hangindent2.3em\hangafter 1
Filipiak, K. and Markiewicz, A. (2003). Optimality of neighbor balanced designs under mixed effects model. {\it Statist. Probab. Lett.} {\bf 61} 225--234. 

\par\noindent\hangindent2.3em\hangafter 1
Filipiak, K. and Markiewicz, A. (2005). Optimality and efficiency of circular neighbor balanced designs for correlated observations. {\it Metrika} {\bf 61} 17--27. 

\par\noindent\hangindent2.3em\hangafter 1
Filipiak, K. and Markiewicz, A. (2007). Optimal designs for a mixed interference model. {\it Metrika} {\bf 65} 369--386.

\par\noindent\hangindent2.3em\hangafter 1
Karlin, S. and Studden, W. (1966).  Tchebycheff systems: With applications in analysis and statistics. {\it Interscience}, New York.

\par\noindent\hangindent2.3em\hangafter 1
Kiefer, J. (1975). Construction and optimality of generalized Youden designs. {\it A Survey of Statistical Design and Linear Models (J. N. Srivistava, ed.)}. North-Holland,
Amsterdam.

\par\noindent\hangindent2.3em\hangafter 1
Kunert, J. (1984). Optimality of balanced uniform repeated measurements designs. {\it Ann. Statist.} {\bf 12} 1006--1017.

\par\noindent\hangindent2.3em\hangafter 1
Kunert, J. and Martin, R. J. (2000). On the determination of optimal designs for an interference model. {{\it Ann. Statist.} {\bf 28} 1728--1742.

\par\noindent\hangindent2.3em\hangafter 1
Kunert, J. and Mersmann, S. (2011). Optimal designs for an interference model. {\it J.Statist. Plann. Inference} {\bf 141} 1623--1632.

\par\noindent\hangindent2.3em\hangafter 1
Kushner, H. B. (1997). Optimal repeated measurements designs: The linear optimality equations. {\it Ann. Statist.} {\bf 25} 2328--2344.

\par\noindent\hangindent2.3em\hangafter 1
Langton, S. (1990). Avoiding edge effects in agroforestry experiments; the use of neighbour-balanced designs and guard areas. {\it Agroforestry Systems} {\bf 12} 173--185.

\par\noindent\hangindent2.3em\hangafter 1
Li, K., Zheng, W. and Ai, M. (2015). Optimal designs for the proportional interference model. {\it Ann. Statist.} {\bf 43} 1596--1616.

\par\noindent\hangindent2.3em\hangafter 1
Morgan, J. P. and Uddin, N. (1991) Two-dimensional design for correlatted errors. {\it Ann. Statist.} {\bf 19} 2160--2182

\par\noindent\hangindent2.3em\hangafter 1
Morgan, J. P. and Uddin, N. (1999) A class of neighbor balanced complete block designs and their efficiencies for spatially correlated errors. {\it Statistics} {\bf 32} 317--330.

\par\noindent\hangindent2.3em\hangafter 1
Uddin, N. and Morgan, J. P. (1997a) Efficient block designs for setting with spatially correlated errors. {\it Biometrika} {\bf 84} 443--454.

\par\noindent\hangindent2.3em\hangafter 1
Uddin, N. and Morgan, J. P. (1997b) Universally optimal designs with blocksize p $\times$ 2 and correlated observations. {\it Ann. Statist.} {\bf 25} 1189--1207 

\par\noindent\hangindent2.3em\hangafter 1
Williams, E. R., John, J. A. and Whitaker, D. (2006). Construction of resolvable spatial row-column designs. {\it Biometrics} {\bf 62} 103--108.

\par\noindent\hangindent2.3em\hangafter 1
Yang, M. (2010). On the de la Garza phenomenon. {\it Ann. Statist.} {\bf 38} 2499--2524. 

\par\noindent\hangindent2.3em\hangafter 1
Yang, M. and Stufken, J. (2009). Support points of locally optimal designs for nonlinear models
with two parameters. {\it Ann. Statist.} {\bf 37} 518--541. 

\par\noindent\hangindent2.3em\hangafter 1
Yang, M. and Stufken, J. (2012)Identifying locally optimal designs for nonlinear models: A simple extension with profound conarrays. {\it Ann. Statist.} {\bf 40} 1665--1681

\par\noindent\hangindent2.3em\hangafter 1
Zheng, W. (2013). Universally optimal crossover designs under subject dropout. {\it Ann. Statist.} {\bf 41} 63--90.

\par\noindent\hangindent2.3em\hangafter 1
Zheng, W. (2015). Universally optimal designs for two interference models. {\it Ann. Statist.} {\bf 43} 501--518.

\par\noindent\hangindent2.3em\hangafter 1
Zheng, W., Ai, M. and Li, K. (2017). Identification of universally optimal circular designs for the interference model. {\it Ann. Statist.} Preprint.

\end{document}